\documentclass[10pt, conference, compsocconf]{IEEEtran}

\usepackage{booktabs} 
\usepackage[english]{babel}
\usepackage{listings}
\usepackage[table,xcdraw]{xcolor}
\usepackage{array}
\usepackage{multirow}
\usepackage{amsmath,amsfonts, scalerel}
\usepackage{amsthm}
\usepackage{textcomp}
\usepackage{balance}
\usepackage{caption}
\usepackage{tikz}
\usepackage{pgfplots}
\pgfplotsset{compat=1.18}
\usetikzlibrary{trees, positioning, mindmap}
\usepackage{lipsum,adjustbox}
\usepackage[multiple]{footmisc}
\usepackage{graphicx}
\usepackage{cite}
\usepackage{algpseudocode}
\usepackage{boldline} 
\usepackage{color, colortbl}

\usepackage[colorlinks,breaklinks]{hyperref} 
\hypersetup{
    hypertexnames=true,
    linkcolor=blue,
    anchorcolor=black,
    citecolor=magenta,
    urlcolor=alizarin
}
\usepackage{cleveref}

\usepackage{multicol}
\usepackage{footnote}
\makesavenoteenv{tabular}
\makesavenoteenv{table}
\usepackage{url}
\definecolor{Gray}{gray}{0.9}
\definecolor{azure(web)(azuremist)}{rgb}{0.94, 1.0, 1.0}
\usepackage{float}
\usepackage[T1]{fontenc}
\usepackage[skins]{tcolorbox}
\newtcolorbox{myframe}[2][]{%
	enhanced,colback=white,colframe=black,coltitle=black,
	sharp corners,boxrule=0.6pt,
	fonttitle=\itshape,
	attach boxed title to top left={yshift=-0.3\baselineskip-0.4pt,xshift=2mm},
	boxed title style={tile,size=minimal,left=0.5mm,right=0.5mm,
		colback=white,before upper=\strut},
	title=#2,#1
}
\usepackage{boldline} 
\usepackage{color, colortbl}

\definecolor{Gray}{gray}{0.9}
\definecolor{airforceblue}{rgb}{0.36, 0.54, 0.66}
\definecolor{aliceblue}{rgb}{0.94, 0.97, 1.0}
\definecolor{alizarin}{rgb}{0.82, 0.1, 0.26}
\definecolor{amber}{rgb}{1.0, 0.75, 0.0}
\definecolor{amber(sae/ece)}{rgb}{1.0, 0.49, 0.0}
\definecolor{antiquebrass}{rgb}{0.8, 0.58, 0.46}
\definecolor{azure(web)(azuremist)}{rgb}{0.94, 1.0, 1.0}
\definecolor{bronze}{rgb}{0.8, 0.5, 0.2}
\definecolor{battleshipgrey}{rgb}{0.52, 0.52, 0.51}
\definecolor{bole}{rgb}{0.47, 0.27, 0.23}
\definecolor{bulgarianrose}{rgb}{0.28, 0.02, 0.03}
\definecolor{cadet}{rgb}{0.33, 0.41, 0.47}
\definecolor{ceil}{rgb}{0.57, 0.63, 0.81}
\definecolor{cerulean}{rgb}{0.0, 0.48, 0.65}
\definecolor{charcoal}{rgb}{0.21, 0.27, 0.31}
\definecolor{coolblack}{rgb}{0.0, 0.18, 0.39}
\definecolor{coolgrey}{rgb}{0.55, 0.57, 0.67}
\definecolor{darkcandyapplered}{rgb}{0.64, 0.0, 0.0}
\definecolor{darkbrown}{rgb}{0.4, 0.26, 0.13}
\definecolor{darkcerulean}{rgb}{0.03, 0.27, 0.49}
\definecolor{darkgray}{rgb}{0.66, 0.66, 0.66}
\definecolor{darkjunglegreen}{rgb}{0.1, 0.14, 0.13}
\definecolor{darktaupe}{rgb}{0.28, 0.24, 0.2}
\definecolor{davy\'sgrey}{rgb}{0.33, 0.33, 0.33}
\definecolor{frenchblue}{rgb}{0.0, 0.45, 0.73}
\definecolor{almond}{rgb}{0.94, 0.87, 0.8}
\definecolor{beaublue}{rgb}{0.74, 0.83, 0.9}
\definecolor{beige}{rgb}{0.96, 0.96, 0.86}
\definecolor{bisque}{rgb}{1.0, 0.89, 0.77}
\definecolor{black}{rgb}{0.0, 0.0, 0.0}
\definecolor{fluorescentorange}{rgb}{1.0, 0.75, 0.0}
\definecolor{ghostwhite}{rgb}{0.97, 0.97, 1.0}
\definecolor{antiquewhite}{rgb}{0.98, 0.92, 0.84}

\newtheorem{definition}{Definition}
\newtheorem{myAttack}{Attack}

\newtheorem{proposition}{Proposition}

\IEEEoverridecommandlockouts
\begin{document}

\pagenumbering{arabic}
\newcommand{\protocol}{PervPPML}
\newcommand{\twoPervPPML}{2PervPPML}
\newcommand{\threePervPPML}{3PervPPML}
\newcommand{\allPPML}{2PervPPML/3PervPPML}
\newcommand{\trustedPervPPML}{TrustedPervPPML}
\newcommand{\protocolTTT}{\texttt{PervPPML}}
\newcommand{\twoPervPPMLTTT}{\texttt{2PervPPML}}
\newcommand{\threePervPPMLTTT}{\texttt{3PervPPML}}
\newcommand{\trustedPervPPMLTTT}{\texttt{TrustedPervPPML}}
\newcommand{\ecgPPML}{\texttt{ecgPPML}}

\makeatletter
\newcommand{\linebreakand}{%
  \end{@IEEEauthorhalign}
  \hfill\mbox{}\par
  \mbox{}\hfill\begin{@IEEEauthorhalign}
}
\makeatother

\title{A Pervasive, Efficient and Private Future: Realizing Privacy-Preserving Machine Learning Through Hybrid Homomorphic Encryption}

\author{\IEEEauthorblockN{Khoa Nguyen}
\IEEEauthorblockA{Department of Computing Sciences \\
Tampere University\\
Tampere, Finland \\
\href{mailto:khoa.nguyen@tuni.fi}{khoa.nguyen@tuni.fi}}
\and
\IEEEauthorblockN{Mindaugas Budzys}
\IEEEauthorblockA{Department of Computing Sciences \\
Tampere University\\
Tampere, Finland \\
\href{mailto:mindaugas.budzys@tuni.fi}{mindaugas.budzys@tuni.fi}}
\and
\IEEEauthorblockN{Eugene Frimpong}
\IEEEauthorblockA{Department of Computing Sciences \\
Tampere University\\
Tampere, Finland \\
\href{mailto:eugene.frimpong@tuni.fi}{eugene.frimpong@tuni.fi}}
\linebreakand
\IEEEauthorblockN{Tanveer Khan}
\IEEEauthorblockA{Department of Computing Sciences \\
Tampere University\\
Tampere, Finland \\
\href{mailto:tanveer.khan@tuni.fi}{tanveer.khan@tuni.fi}}
\and
\IEEEauthorblockN{Antonis Michalas}
\IEEEauthorblockA{Department of Computing Sciences \\
Tampere University, Tampere, Finland and\\
RISE Research Institutes of Sweden \\
\href{mailto:antonios.michalas@tuni.fi}{antonios.michalas@tuni.fi}}
}

\maketitle
\thispagestyle{plain}
\pagestyle{plain}
\begin{abstract}
Machine Learning (ML) has become one of the most impactful fields of data science in recent years. However, a significant concern with ML is its privacy risks due to rising attacks against ML models. Privacy-Preserving Machine Learning (PPML) methods have been proposed to mitigate the privacy and security risks of ML models. A popular approach to achieving PPML uses Homomorphic Encryption (HE). However, the highly publicized inefficiencies of HE make it unsuitable for highly scalable scenarios with resource-constrained devices. Hence, Hybrid Homomorphic Encryption (HHE) -- a modern encryption scheme that combines symmetric cryptography with HE -- has recently been introduced to overcome these challenges. HHE potentially provides a foundation to build new efficient and privacy-preserving services that transfer expensive HE operations to the cloud. This work introduces HHE to the ML field by proposing resource-friendly PPML protocols for edge devices. More precisely, we utilize HHE as the primary building block of our PPML protocols. We assess the performance of our protocols by first extensively evaluating each party's communication and computational cost on a dummy dataset and show the efficiency of our protocols by comparing them with similar protocols implemented using plain BFV. Subsequently, we demonstrate the real-world applicability of our construction by building an actual PPML application that uses HHE as its foundation to classify heart disease based on sensitive ECG data. 
\end{abstract}


\section{Introduction}
\label{sec: introduction}

In less than a decade, ML has been transformed from a quiet backwater of computing science into the hottest and most hyped area of science, allowing various users to classify and make predictions on multi-dimensional data. 
However, to achieve this, the results should be accompanied by a large amount of high-quality training data requiring the collaboration of several organizations. 
Regulations such as the 
GDPR forbid the sharing and processing of sensitive data without the data subject's consent. It has, therefore, become crucial to uphold data privacy, confidentiality, and profit-sharing while obtaining data from third parties. 
One solution to this problem is employing Privacy-Preserving Machine Learning (PPML). PPML ensures that the use of data protects user privacy and that data is utilized in a safe fashion, avoiding leakage of confidential and private information. To this end, researchers have proposed and implemented various PPML-achieving techniques, varying from secure cryptographic schemes to distributed, hybrid, and data modification approaches. This work focuses on cryptographic approaches, the most commonly used one being Homomorphic Encryption (HE)~\cite{rivest1978data,gentry2009fully}, which exhibits a high potential in ML applications such as Machine-Learning-as-a-Service (MLaaS). In MLaaS applications, data is encrypted under an HE scheme and transferred to the cloud for processing. Furthermore, due to the versatility of HE, it can also be used to keep the model private through model parameter encryption. However, despite the various advances in HE, it has yet to find mainstream use because of high computational complexity and extended ciphertext expansion resulting in very large ciphertexts -- the transmission of which becomes a significant bottleneck in practice~\cite{dobraunig2021pasta}. To address these issues, researchers have turned to the Hybrid Homomorphic Encryption (HHE) concept~\cite{dobraunig2021pasta,bakas2019modern}, which uses symmetric ciphers to make HE more accessible for users.

\begin{tcolorbox}[colback=white!10!white,
                     colframe=white!20!black,
                     title=\textsc{\LARGE\texttt{HHE in a Nutshell}},  
                     center, 
                     valign=top, 
                     halign=left,
                     before skip=0.3cm, 
                     after skip=0.6cm,
                     center title]

In an HHE scheme, users locally encrypt data via a symmetric key encryption scheme rather than encrypt data with only an HE scheme. Subsequently, they homomorphically encrypt the symmetric key used in
\end{tcolorbox}
\begin{tcolorbox}[colback=white!10!white,
                     colframe=white!20!black,
                    ]
the encryption process with an HE scheme. The two ciphertexts resulting from the above-mentioned encryptions are then forwarded to the server. Upon receiving both ciphertexts, the server transforms the symmetric ciphertext into a homomorphic ciphertext using the homomorphically encrypted symmetric key. The use of HHE produces ciphertexts of a substantially smaller size compared to using only HE schemes because of the symmetric encryption. The reduced size drastically lowers the communication overhead a user would typically incur using only a traditional HE approach. However, due to high multiplicative depth, not all symmetric key encryption schemes are compatible with HHE. To this end, researchers have designed a number of HE-friendly symmetric ciphers for use in HHE schemes, such as HERA/Rubato~\cite{cho2021transciphering, ha2022rubato}, Elisabeth~\cite{cosseron2022towards} and PASTA~\cite{dobraunig2021pasta}. Each of these ciphers has noticeable differences, which we discuss in \autoref{sec:relatedwork}. Aside from reducing the size of ciphertext, HHE also provides a way for most consumer-grade devices to benefit from the advantages of HE schemes by transferring the more computationally expensive operations to the server or cloud.
\end{tcolorbox}


For this work, we adopt the concept of HHE, as implemented in PASTA~\cite{dobraunig2021pasta}, to address the problems of PPML, which are not solved by using only HE. To achieve this, we expand upon a previous work, named GuardML~\cite{10.1145/3605098.3635983}, by providing a new construction making use of a Trusted Execution Environment (TEE) as well as expanding upon the threat model and experimental results. Apart from the novelty of incorporating HHE in PPML, we also aim to smooth out the big hurdles of implementing strong PPML models in a wide range of devices.

\smallskip



This paper makes the following technical contributions: 
\begin{enumerate}
	\item We show how the novel concept of HHE can be used effectively to address the problem of PPML. This is accomplished by expanding the use case of PASTA to ML applications. Apart from the novelty of incorporating HHE in PPML, our approach's main aim is to smooth out the big hurdles of implementing strong PPML models in a wide range of devices. We believe that this might spur PPML into a new and more interesting territory that has the potential to provide a very careful balance between functionality and privacy, especially in the field of pervasive computing where, in many cases, the underlying infrastructure has certain inbuilt limitations.
	
    \item We formally design two protocols that allow an authorized entity (e.g.\ an analyst) to process encrypted data efficiently \textit{as if} they were unencrypted. In one of our proposed protocols, we utilize a Trusted Execution Environment (TEE) to achieve a completely trusted architecture. By using a TEE, our protocol supports a secure and integrity-protected processing environment for key operations separate from the other participating entities. Subsequently, we prove the security of our protocols by showing that there is no leakage of information that could potentially breach user privacy.
	
    \item Through extensive experiments, we show that our approach has lower computational and communications overhead when compared to existing approaches based on the traditional HE concept without sacrificing security, privacy, and accuracy. We also argue that our solution applies to a wide range of devices  -- a result we believe can prove to be fundamental for many applications. To this end, our protocol's applicability to a real-world ML scenario with a sensitive medical dataset is demonstrated. Experimental results show that compared to the plaintext version, our protocol for PPML achieves almost comparable accuracy results while preserving both the privacy of the dataset and the neural network. Furthermore, the majority of computation cost is outsourced to the CSP.
 
\end{enumerate}

\section{Related Works}
\label{sec:relatedwork}

Various recent works have proposed using HE schemes in implementing PPML~\cite{khan2024learning,khan2021blind,khan2023love,khan2023more,khan2023split}. Gentry's work~\cite{gentry2009fully} revolutionized the field of HE and paved the way for multiple modern schemes, such as TFHE~\cite{chillotti2020tfhe}, BFV~\cite{fan2012somewhat}, and CKKS~\cite{cheon2017homomorphic} in PPML applications. Each HE scheme introduces different advantages and limitations in ML. Hence, it is important to understand their differences (\autoref{tab:comp-he-schemes}). These schemes have achieved high-accuracy results in PPML. Examples of TFHE-based PPML works are TAPAS~\cite{sanyal2018tapas}, FHE-DiNN~\cite{bourse2018fast} and Glyph~\cite{lou2020glyph}. An example of a CKKS-based PPML protocol is POSEIDON~\cite{sav2020poseidon}, while an example of a BFV-based PPML work is HCNN~\cite{al2020towards}. HE schemes suffer from large ciphertext sizes and high computational complexities, which make them unsuitable for every environment~\cite{khan2024wildest}.

\begin{table*}[ht!]
\centering
\resizebox{.75\linewidth}{!}{%
\begin{tabular}{ |>{\columncolor{Gray}}c|c|c|c| } 
 \hlineB{4}
 \rowcolor{Gray} 
 \cellcolor{azure(web)(azuremist)} & \textbf{CKKS}~\cite{cheon2017homomorphic} & \textbf{TFHE}~\cite{chillotti2016faster} & \textbf{BFV}~\cite{brakerski2012fully} \\ 
 \hlineB{3}
 \cellcolor{Gray}\textbf{Data Type} & Floating point & Binary & Integer \\ 
 \hline
 \cellcolor{Gray}\textbf{HE Type} & Somewhat HE & Fully HE & Somewhat HE\\ 
 \hline
 \cellcolor{Gray}\textbf{Batching} & Yes & No* & Yes \\
 \hline
 \cellcolor{Gray}\textbf{Non-linear Activation Functions} & Polynomial Approximation & Look-up Tables & Polynomial Approximation \\ 
 \hline
 \cellcolor{Gray}\textbf{Precision} & Highest & Lower than CKKS & Lower than CKKS \\ 
 \hline
 \cellcolor{Gray}\textbf{Memory usage} & Low & High & Mid \\
 \hline
 \cellcolor{Gray}\textbf{Efficiency} & Mid & High & Low \\ 
 \hlineB{3}
\end{tabular}
}
\caption{\centering Comparison of different HE schemes. * -- TFHE can implement horizontal or vertical packing, which functions similarly to BFV and CKKS batching but is computed differently.}
\label{tab:comp-he-schemes}
\end{table*}

In response to these challenges, the first approaches to designing HHE schemes relied on existing and well-established symmetric ciphers such as AES~\cite{gentry2012homomorphic}. However, AES has since been proven not to be a good fit for HHE schemes, primarily due to its large multiplicative depth~\cite{dobraunig2021pasta}. Thus research in the field of HHE took a different approach, where the main focus has been shifted to the design of symmetric ciphers with different optimization criteria~\cite{bakas2022symmetrical}, such as eliminating the ciphertext expansion~\cite{canteaut2018stream} or using filter permutators~\cite{meaux2019improved}.  
However, to date, HHE has seen limited practical application~\cite{bakas2022symmetrical} in real-world applications, and only a handful of works exist in the field of PPML. To the best of our knowledge, the \textit{main} HHE schemes currently are HERA~\cite{cho2021transciphering}, Elisabeth~\cite{cosseron2022towards} and PASTA~\cite{dobraunig2021pasta}. The authors of HERA also proposed Rubato~\cite{ha2022rubato}; however, the specifications remain largely the same as in HERA. These proposed approaches have different specifications and can be applied to different use cases. We provide an overview of these approaches in \autoref{tab:comp-hhe-schemes}. HERA~\cite{cho2021transciphering} is a stream cipher based on the CKKS HE scheme and allows computations on floating point data types. In comparison, Elisabeth~\cite{cosseron2022towards} is designed to utilize the TFHE scheme, while PASTA~\cite{dobraunig2021pasta} is based on BFV for integer data types. HERA and PASTA are defined over $\mathbb{Z}_{\mathsf{q}}$, where $\mathsf{q} = 2^{16} + 1$, and can store up to 16-bit inputs. 
Meanwhile, Elisabeth is defined over $\mathbb{Z}_{\mathsf{q}}$, where $\mathsf{q} = 2^4$, and can store up to 4 bits of data.

\begin{table*}[ht!]
\centering
\begin{tabular}{ |>{\columncolor{Gray}}c|c|c|c| } 
 \hlineB{4}
 \rowcolor{Gray}
 \cellcolor{azure(web)(azuremist)} & \textbf{HERA}~\cite{cho2021transciphering} & \textbf{Elisabeth}~\cite{cosseron2022towards} & \textbf{PASTA}~\cite{dobraunig2021pasta} \\ 
 \hlineB{3}
 \cellcolor{Gray}\textbf{Programming Language} & Go & Rust & C \\ 
 \hline
 \cellcolor{Gray}\textbf{Data Type} & Floating point & Binary & Integer \\ 
 \hline
 \cellcolor{Gray}\textbf{HE Scheme} & CKKS & TFHE & BFV\\ 
 \hline
 \cellcolor{Gray}\textbf{Application in ML} & No & Yes & No \\
 \hline
 \cellcolor{Gray}\textbf{Defined over} & $\mathbb{Z}_\mathsf{q}$, where $\mathsf{q} > 2^{16}$ & $\mathbb{Z}_\mathsf{q}$, where $\mathsf{q} = 2^{4}$ &$\mathbb{F}_\mathsf{p}$, where $\mathsf{p}$ is a 16-bit prime \\ 
 \hline
 \cellcolor{Gray}\textbf{Security Level} & 80 or 128-bit & 128-bit & 128-bit \\ 
 \hline
 \cellcolor{Gray}\textbf{Quantization} & No & Yes & Yes \\ 
 \hlineB{3}
\end{tabular}
\caption{Comparison of different HHE schemes}
\label{tab:comp-hhe-schemes}
\end{table*}

Each approach achieves the same security level of~128-bits. Additionally, HERA also provides tests for a security level of 80 bits. As HERA allows computations on floating point data types, it does not require quantization on certain inputs. Elisabeth and PASTA, on the other hand, require quantization to operate on floating point numbers, which introduces a rounding error, which can reduce the accuracy of certain applications, i.e.\ ML. To the best of our knowledge, HERA has yet to be applied to an ML application. One of the use cases provided by Elisabeth is a CNN classification task on the Fashion-MNIST dataset and shows equivalent accuracy~($84.18\%$) to the cleartext model without HE. 
For this work, we adopt the PASTA HHE scheme as the primary building block for our protocols based on the following reasons:
\begin{itemize}
\item We compared the implementations of PASTA and HERA by utilizing the available open-source code\footnote{HERA: \url{https://github.com/KAIST-CryptLab/RtF-Transciphering}}\textsuperscript{,}\footnote{PASTA: \url{https://github.com/IAIK/hybrid-HE-framework}} on our systems. Our findings indicate that HERA requires substantial memory, with RAM usage reaching approximately 60GB for 80-bit security. Conversely, PASTA exhibited significantly lower memory requirements while achieving 128-bit security, which led to our decision to use PASTA over HERA. 
 
\item We aim to evaluate the suitability of PASTA when applied to PPML and investigate the feasibility of applying PASTA to various real-world ML datasets, particularly medical datasets, with the potential to achieve comparable accuracy and higher efficiency than conventional HE techniques. 

\end{itemize}
	
\color{black}
\section{Hybrid Homomorphic Encryption}
\label{subsec: hhe}

\begin{definition}[Hybrid Homomorphic Encryption] Let $\mathsf{HE}$ be a Homomorphic Encryption scheme and $\mathsf{SKE} = (\mathsf{Gen, Enc, Dec})$ be a symmetric-key encryption scheme. Moreover, let $\mathcal{X} = (x_1, \dots, x_n)$ be the message space and $\lambda$ the security parameter. An $\mathsf{HHE}$ scheme 
consists of five PPT algorithms $\mathsf{HHE = (KeyGen, Enc, Decomp,}$ \allowbreak $\mathsf{ Eval, Dec)}$ such that: 
\end{definition}

\begin{itemize}
		\item $\mathbf{HHE.KeyGen}$: The key generation algorithm takes as input a security parameter $\lambda$ and outputs a HE public/private key pair ($\mathsf{pk}$/$\mathsf{sk}$) and a HE evaluation key ($\mathsf{evk}$).

		\item $\mathbf{HHE.Enc}$: The encryption algorithm consists of three steps:
		\begin{itemize}
			\item $\mathsf{SKE.Gen}$: The SKE generation algorithm takes as input the security parameter $\lambda$ and outputs a symmetric key $\mathsf{K}$.
			
                \item $\mathsf{HE.Enc}$: An HE encryption algorithm that takes as input $\mathsf{pk}$ and $\mathsf{K}$, and outputs $c_\mathsf{K}$ -- a homomorphically encrypted representation of the symmetric key $\mathsf{K}$.
			\item $\mathsf{SKE.Enc}$: The SKE encryption algorithm takes as input a message $x$ and $\mathsf{K}$ and outputs a ciphertext $c$.
		\end{itemize}
		\item $\mathbf{HHE.Decomp}$: This algorithm takes as an input the evaluation key $\mathsf{evk}$, the symmetrically encrypted ciphertext $c$, and the homomorphically encrypted symmetric key $c_\mathsf{K}$, and outputs $c'$ -- a homomorphic encryption of the original message $x$. 

		\item $\mathbf{HHE.Eval}$: This algorithm takes as input $n$ homomorphic ciphertexts $c'_n$, where $n \geq 2$, the evaluation key $\mathsf{evk}$ and a homomorphic function $f$, and outputs a ciphertext $c'_{eval}$ of the evaluation results.

		\item $\mathbf{HHE.Dec}$: The decryption algorithm takes as input a private key $\mathsf{sk}$ and the evaluated ciphertext $c'_{eval}$ and outputs $f(x)$.  

	\end{itemize} 

The correctness of an HHE scheme follows directly from the correctness of the underlying public-key HE scheme.

\section{System Model}
\label{sec:architecture}
In this section, we introduce our system model by explicitly describing our protocol's 
main entities and 
their capabilities. 
\begin{itemize}
	\item \textbf{User}: Let $\mathcal{U} = \{u_1, \ldots, u_n\}$ be the set of all users. Each user generates a unique symmetric key $\mathsf{K_i}$ locally and encrypts their data. Then the generated ciphertexts are outsourced to the $\mathbf{CSP}$ along with an HE encryption $c_\mathsf{K_i}$ of the underlying symmetric key.

	\item \textbf{Cloud Service Provider (CSP)}: 
 Primarily responsible for gathering symmetrically encrypted data from multiple users. The $\mathbf{CSP}$ is tasked with converting the symmetrically encrypted data into homomorphic ciphertexts and, upon request, performing blind operations on them. 

    \item \textbf{Trusted Execution Environment (TEE)}: In one of our proposed protocols, we utilize a \textbf{TEE} on the $\mathbf{CSP}$. In this proposed use case, the \textbf{TEE} is responsible for securely sharing results from the learning phase of our protocol with a requesting party. The specifics of the \textbf{TEE} are beyond the scope of this work; as such, we do not delve into the finer details. A more detailed description is provided in~\cite{tee2015,paju2023sok}. 
    
	\item \textbf{Analyst (A)}: In 
our proposed protocols, there exists an analyst who owns an ML model and is interested in learning the output of ML operations on the encrypted data stored at the $\mathbf{CSP}$. In these protocols, $\mathbf{A}$ decrypts the encrypted data from the HE evaluation of collected encrypted data and, thus, gains insights from user data.
\end{itemize}

\section{PervPPML}
\label{sec:protocol}
In this section, we discuss in detail the construction of \protocolTTT{} -- our Hybrid Homomorphic Privacy-Preserving protocols that constitute the core of this paper. \protocolTTT{} is comprised of two protocols, 
\threePervPPMLTTT{} and \trustedPervPPMLTTT{}. The primary differences between the protocols mentioned above lie in the presence of 
a \textbf{TEE}. 
\threePervPPMLTTT, is ideal for a setting where the analyst is the owner of the ML model and does not wish to reveal the contents to the \textbf{CSP}. 
While \trustedPervPPMLTTT{} would be ideal for a zero-trust setting where each party can collude with another party to break the protocol.   



\noindent\textbf{\textit{Building Blocks:}} 
\label{subsubsec:buildingblocks}
Before proceeding to describe each protocol, we first define the building blocks 
used in our constructions. 

\begin{itemize}
	\item A secure symmetric cipher $\mathsf{SKE = (Gen, Enc, Dec)}$.
	\item A BFV-based HHE scheme $\mathsf{HHE = (KeyGen, Enc, Dec, Decomp,}$$ \allowbreak \mathsf{Eval)}$.
    \item A CCA2 secure public-key encryption scheme $\mathsf{PKE = (Gen,}$\allowbreak$\mathsf{ Enc, Dec)}$.
	\item An EUF-CMA secure signature scheme $\mathsf{\sigma = (sign, ver)}$.    
	\item A first and second pre-image resistant cryptographic hash function $\mathsf{H(\cdot)}$.
\end{itemize}

\subsection{PervPPML: 3-Party Setting}
\label{subsec:3party}
\noindent \textbf{\textit{High-Level Overview:}} \enskip 
\threePervPPMLTTT{} consists of three parties: a set of users $\mathcal{U}$, \textbf{CSP}, and an analyst \textbf{A}. 
The HHE keys are generated by \textbf{A}. 
In the \threePervPPMLTTT{} setting, we assume that an analyst \textbf{A} owns the ML model with parameters $(w, b)$, while a user $u_i$ provides the data $x_i$. \textbf{A} generates the required HHE keys ($\mathsf{pk_{A}, sk_{A}, evk_{A}}$), publishes $\mathsf{pk_{A}}$, and sends $\mathsf{evk_{A}}$ to the \textbf{CSP}. Each $u_{i}$ generates a symmetric key $\mathsf{K_{i}}$ and encrypts their data $x_{i}$ locally to output a symmetric ciphertext $c_{x_i}$. Additionally, $u_{i}$ also generates $c_{\mathsf{K_{i}}}$ -- a homomorphic encryption of the symmetric key $\mathsf{K_{i}}$ using \textbf{A}'s public key and sends both the encrypted data $c_{x_i}$ and $c_{\mathsf{K_{i}}}$ to the \textbf{CSP}. Upon reception, the \textbf{CSP} stores the values locally. In the evaluation phase, \textbf{A} can request a prediction to be performed on the stored $c'_{x_i}$ by first sending to the \textbf{CSP} their homomorphically encrypted pre-trained ML model parameters $(c_{w}, c_{b})$. On receiving the request from \textbf{A}, the \textbf{CSP} transforms $c_{x_i}$ into homomorphic ciphertext $c'_{x_i}$. \textbf{CSP} produces an encrypted result $c_{res}$ and sends the results back to \textbf{A}. Finally, \textbf{A} decrypts the prediction result $res$ using $\mathsf{sk_A}$.

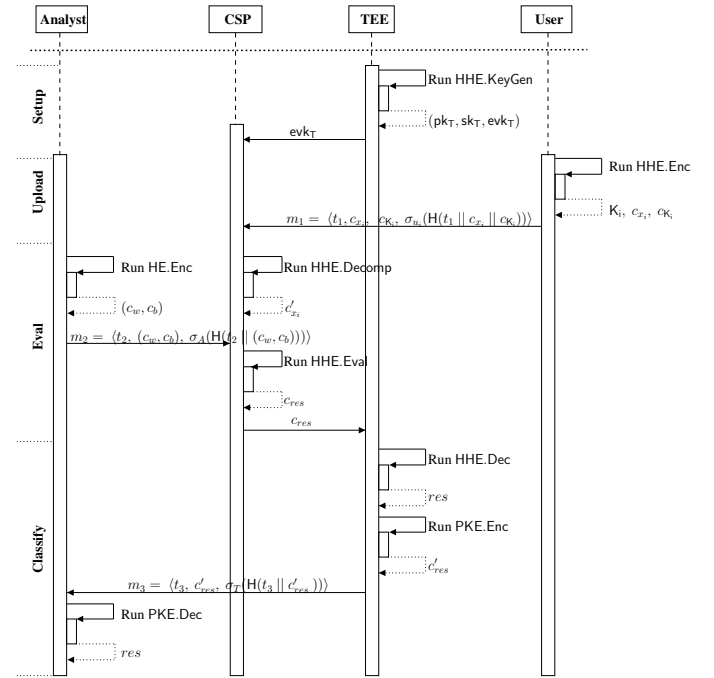
\begin{figure}[ht]
\centering
	\resizebox{.5\textwidth}{!}{%
		\tikzset{every picture/.style={line width=0.75pt}} 
		\begin{tikzpicture}[x=0.75pt,y=0.75pt,yscale=-1,xscale=1]
			\draw [color={rgb, 255:red, 0; green, 0; blue, 0 }  ,draw opacity=1 ][line width=1.5]  [dash pattern={on 1.69pt off 2.76pt}]  (68.5,292) -- (424,292) -- (805.5,291) ;

			\draw  [dash pattern={on 0.84pt off 2.51pt}]  (50.5,312) -- (105,312) ;
			\draw  [dash pattern={on 0.84pt off 2.51pt}]  (50.5,435) -- (99,435) ; 
			\draw  [dash pattern={on 0.84pt off 2.51pt}]  (50.5,547.5) -- (99,547.5) ;
			\draw  [dash pattern={on 0.84pt off 2.51pt}]  (50.5,810) -- (99,810) ;
			\draw  [dash pattern={on 0.84pt off 2.51pt}]  (50.5,1120) -- (99,1120) ;
			
            \draw  [dash pattern={on 4.5pt off 4.5pt}]  (107.5,269) -- (107.5,432) ;
            \draw  [dash pattern={on 4.5pt off 4.5pt}]  (340,269) -- (340,390) ;
            \draw  [dash pattern={on 4.5pt off 4.5pt}]  (520,269) -- (520,314) ;
            \draw  [dash pattern={on 4.5pt off 4.5pt}]  (753.5,269) -- (753.5,432) ;

            \draw   (76,233) -- (144,233) -- (144,268.5) -- (76,268.5) -- cycle ; 
            \draw   (305,233) -- (375,233) -- (375,268.5) -- (305,268.5) -- cycle ; 
			\draw   (490,233) -- (550.06,233) -- (550.06,268.5) -- (490,268.5) -- cycle ;
            \draw   (718,233) -- (785.06,233) -- (785.06,268.5) -- (718,268.5) -- cycle ;
			
			\draw (80,243) node [anchor=north west][inner sep=0.75pt]   [align=left] [font=\Large] {\textbf{Analyst}};
            \draw (324,243) node [anchor=north west][inner sep=0.75pt]   [align=left] [font=\Large] {\textbf{CSP}};
            \draw (503.97,243) node [anchor=north west][inner sep=0.75pt]   [align=left] [font=\Large] {\textbf{TEE}};
            \draw (735,243) node [anchor=north west][inner sep=0.75pt]   [align=left] [font=\Large] {\textbf{User}};

            \draw (70.19,390.44) node [anchor=north west][inner sep=0.75pt]  [rotate=-270.41] [align=left] [font=\Large] {\textbf{Setup}};
			\draw (70.19,510.44) node [anchor=north west][inner sep=0.75pt]  [rotate=-270.41] [align=left] [font=\Large] {\textbf{Upload}};
			\draw (70.19,690.44) node [anchor=north west][inner sep=0.75pt]  [rotate=-270.41] [align=left] [font=\Large] {\textbf{Eval}};
			\draw (70.19,985.44) node [anchor=north west][inner sep=0.75pt]  [rotate=-270.41] [align=left] [font=\Large] {\textbf{Classify}};

            \draw   (99,430) -- (116.5,430) -- (116.5,1120) -- (99,1120) -- cycle ;
            \draw   (333,390) -- (350.5,390) -- (350.5,1120) -- (333,1120) -- cycle ;
            \draw   (512,312) -- (529.5,312) -- (529.5,1120) -- (512,1120) -- cycle ;
            \draw   (745,430) -- (762.5,430) -- (762.5,1120) -- (745,1120) -- cycle ;

            \draw    (530,318) -- (591.5,318) ;
			\draw    (591.5,318) -- (591.5,339) ;
			\draw    (545.5,339) -- (591.5,339) ;
			\draw (593.5,321) node [anchor=north west][inner sep=0.75pt]   [align=left] [font=\Large] {Run$\displaystyle \ \mathsf{HHE.KeyGen}$};

            \draw [shift={(543.5,339)}, rotate = 0] [fill={rgb, 255:red, 0; green, 0; blue, 0 }  ][line width=0.08]  [draw opacity=0] (8.93,-4.29) -- (0,0) -- (8.93,4.29) -- cycle    ; 
			\draw   (529.5,339) -- (542.5,339) -- (542.5,372) -- (529.5,372) -- cycle ;      

            \draw  [dash pattern={on 0.84pt off 2.51pt}]  (541.5,372) -- (591.5,372) ;
			\draw  [dash pattern={on 0.84pt off 2.51pt}]  (591.5,372) -- (591.5,393) ;
			\draw  [dash pattern={on 0.84pt off 2.51pt}]  (529.5,392) -- (591.5,392) ;
   
			\draw [shift={(529.5,392)}, rotate = 0] [fill={rgb, 255:red, 0; green, 0; blue, 0 }  ][line width=0.08]  [draw opacity=0] (8.93,-4.29) -- (0,0) -- (8.93,4.29) -- cycle    ; 
			
			\draw (593.5,375) node [anchor=north west][inner sep=0.75pt]   [align=left] [font=\Large] {$\displaystyle (\mathsf{pk_T ,sk_T ,evk_T}) \ $};

            \draw    (350,410) -- (512,410) ;
            \draw [shift={(350,410)}, rotate = 0] [fill={rgb, 255:red, 0; green, 0; blue, 0 }  ][line width=0.08]  [draw opacity=0] (8.93,-4.29) -- (0,0) -- (8.93,4.29) -- cycle    ;

            \draw (410.5,390) node [anchor=north west][inner sep=0.75pt]   [align=left] [font=\Large] {$\displaystyle \mathsf{evk_T} \ $};

            \draw    (763,435) -- (824.5,435) ;
			\draw    (824.5,435) -- (824.5,456) ; 
			\draw    (782.5,456) -- (824.5,456) ;
			\draw [shift={(776,456)}, rotate = 0] [fill={rgb, 255:red, 0; green, 0; blue, 0 }  ][line width=0.08]  [draw opacity=0] (8.93,-4.29) -- (0,0) -- (8.93,4.29) -- cycle    ;
			\draw   (763,456) -- (776,456) -- (776,489) -- (763,489) -- cycle ;
			\draw  [dash pattern={on 0.84pt off 2.51pt}]  (776,489) -- (824.5,489) ;
			\draw  [dash pattern={on 0.84pt off 2.51pt}]  (824.5,489) -- (824.5,510) ;
			\draw  [dash pattern={on 0.84pt off 2.51pt}]  (824.5,510) -- (766,510) ;
			\draw [shift={(763,510)}, rotate = 0] [fill={rgb, 255:red, 0; green, 0; blue, 0 }  ][line width=0.08]  [draw opacity=0] (8.93,-4.29) -- (0,0) -- (8.93,4.29) -- cycle    ;

            \draw (833,437) node [anchor=north west][inner sep=0.75pt]   [align=left] [font=\Large] {Run$\displaystyle \ \mathsf{HHE.Enc}$};

            \draw (833,495) node [anchor=north west][inner sep=0.75pt]   [align=left] [font=\Large] {$\displaystyle \mathsf{K_{i}} ,\ c_{x_i},\ c_\mathsf{K_i}$};

            \draw    (350,525) -- (744.5,525) ;
            \draw [shift={(350,525)}, rotate = 0] [fill={rgb, 255:red, 0; green, 0; blue, 0 }  ][line width=0.08]  [draw opacity=0] (8.93,-4.29) -- (0,0) -- (8.93,4.29) -- cycle    ;

            \draw (406,505) node [anchor=north west][inner sep=0.75pt]  [font=\Large]  {$m_{1} =\ \langle t_{1}, c_{x_i},\ \ c_\mathsf{K_i} ,\ \sigma _{u_{i}}(\mathsf{H}( t_{1} \ || \ c_{x_i}\ || \ c_\mathsf{K_i})) \rangle $}; 

            \draw    (350,565) -- (401.5,565) ;
			\draw    (401.5,565) -- (401.5,586) ;
			\draw    (360.5,586) -- (401.5,586) ;
			\draw (402.5,570) node [anchor=north west][inner sep=0.75pt]   [align=left] [font=\Large] {Run$\displaystyle \ \mathsf{HHE.Decomp}$};

            \draw [shift={(360.5,586)}, rotate = 0] [fill={rgb, 255:red, 0; green, 0; blue, 0 }  ][line width=0.08]  [draw opacity=0] (8.93,-4.29) -- (0,0) -- (8.93,4.29) -- cycle    ; 
			\draw   (350.5,586) -- (363.5,586) -- (363.5,619) -- (350.5,619) -- cycle ;      

            \draw  [dash pattern={on 0.84pt off 2.51pt}]  (363.5,619) -- (401.5,619) ;
			\draw  [dash pattern={on 0.84pt off 2.51pt}]  (401.5,619) -- (401.5,640) ;
			\draw  [dash pattern={on 0.84pt off 2.51pt}]  (350.5,640) -- (401.5,640) ;

            \draw [shift={(350.5,640)}, rotate = 0] [fill={rgb, 255:red, 0; green, 0; blue, 0 }  ][line width=0.08]  [draw opacity=0] (8.93,-4.29) -- (0,0) -- (8.93,4.29) -- cycle    ; 

            \draw (404,622) node [anchor=north west][inner sep=0.75pt]   [align=left] [font=\Large] {$\displaystyle c'_{x_i}\ $}; 

            \draw    (117,565) -- (178.5,565) ; 
			\draw    (178.5,565) -- (178.5,586) ;
			\draw    (132.5,586) -- (178.5,586) ;
			\draw [shift={(129.5,586)}, rotate = 0] [fill={rgb, 255:red, 0; green, 0; blue, 0 }  ][line width=0.08]  [draw opacity=0] (8.93,-4.29) -- (0,0) -- (8.93,4.29) -- cycle    ;
			\draw   (116.5,586) -- (129.5,586) -- (129.5,619) -- (116.5,619) -- cycle ;
			\draw  [dash pattern={on 0.84pt off 2.51pt}]  (128.5,619) -- (180.5,619) ;
			\draw  [dash pattern={on 0.84pt off 2.51pt}]  (180.5,619) -- (180.5,640) ;
			\draw  [dash pattern={on 0.84pt off 2.51pt}]  (120.5,640) -- (179.5,640) ;
			\draw [shift={(117.5,640)}, rotate = 0] [fill={rgb, 255:red, 0; green, 0; blue, 0 }  ][line width=0.08]  [draw opacity=0] (8.93,-4.29) -- (0,0) -- (8.93,4.29) -- cycle    ;

            \draw (187,570) node [anchor=north west][inner sep=0.75pt]   [align=left] [font=\Large] {Run$\displaystyle \ \mathsf{HE.Enc}$};

            \draw (187,622) node [anchor=north west][inner sep=0.75pt]   [align=left] [font=\Large] {$\displaystyle (c_{w}, c_{b})$};

            \draw    (117,680) -- (332.5,680) ;
            \draw [shift={(332.5,680)}, rotate = 180] [fill={rgb, 255:red, 0; green, 0; blue, 0 }  ][line width=0.08]  [draw opacity=0] (8.93,-4.29) -- (0,0) -- (8.93,4.29) -- cycle    ; 

            \draw (120,660.4) node [anchor=north west][inner sep=0.75pt]  [font=\Large]  {$m_{2}=\ \langle t_{2},\ (c_{w}, c_{b}),\ \sigma _{A}( \mathsf{H}( t_{2} \ || \ (c_{w}, c_{b}))) \rangle $}; 

            \draw    (350,690) -- (401.5,690) ;
			\draw    (401.5,690) -- (401.5,711) ;
			\draw    (360.5,711) -- (401.5,711) ;
			\draw (402.5,695) node [anchor=north west][inner sep=0.75pt]   [align=left] [font=\Large] {Run$\displaystyle \ \mathsf{HHE.Eval}$};

            \draw [shift={(360.5,711)}, rotate = 0] [fill={rgb, 255:red, 0; green, 0; blue, 0 }  ][line width=0.08]  [draw opacity=0] (8.93,-4.29) -- (0,0) -- (8.93,4.29) -- cycle    ; 
			\draw   (350.5,711) -- (363.5,711) -- (363.5,744) -- (350.5,744) -- cycle ;      

            \draw  [dash pattern={on 0.84pt off 2.51pt}]  (363.5,744) -- (401.5,744) ;
			\draw  [dash pattern={on 0.84pt off 2.51pt}]  (401.5,744) -- (401.5,765) ;
			\draw  [dash pattern={on 0.84pt off 2.51pt}]  (350.5,765) -- (401.5,765) ;

            \draw [shift={(350.5,765)}, rotate = 0] [fill={rgb, 255:red, 0; green, 0; blue, 0 }  ][line width=0.08]  [draw opacity=0] (8.93,-4.29) -- (0,0) -- (8.93,4.29) -- cycle    ; 

            \draw (402.5,750) node [anchor=north west][inner sep=0.75pt]   [align=left] [font=\Large] {$\displaystyle c_{res}$}; 

            \draw    (350,795) -- (512,795) ;
            \draw [shift={(512,795)}, rotate = 180] [fill={rgb, 255:red, 0; green, 0; blue, 0 }  ][line width=0.08]  [draw opacity=0] (8.93,-4.29) -- (0,0) -- (8.93,4.29) -- cycle    ;

            \draw (412,777) node [anchor=north west][inner sep=0.75pt]   [align=left] [font=\Large] {$\displaystyle c_{res} \ $};

            \draw    (530,820) -- (591.5,820) ;
			\draw    (591.5,820) -- (591.5,841) ;
			\draw    (545.5,841) -- (591.5,841) ;
			\draw (593.5,825) node [anchor=north west][inner sep=0.75pt]   [align=left] [font=\Large] {Run$\displaystyle \ \mathsf{HHE.Dec}$};

            \draw [shift={(543.5,841)}, rotate = 0] [fill={rgb, 255:red, 0; green, 0; blue, 0 }  ][line width=0.08]  [draw opacity=0] (8.93,-4.29) -- (0,0) -- (8.93,4.29) -- cycle    ; 
			
            \draw   (529.5,841) -- (542.5,841) -- (542.5,874) -- (529.5,874) -- cycle ;      

            \draw  [dash pattern={on 0.84pt off 2.51pt}]  (541.5,874) -- (591.5,874) ;
			\draw  [dash pattern={on 0.84pt off 2.51pt}]  (591.5,874) -- (591.5,895) ;
			\draw  [dash pattern={on 0.84pt off 2.51pt}]  (529.5,895) -- (591.5,895) ;
   
			\draw [shift={(529.5,895)}, rotate = 0] [fill={rgb, 255:red, 0; green, 0; blue, 0 }  ][line width=0.08]  [draw opacity=0] (8.93,-4.29) -- (0,0) -- (8.93,4.29) -- cycle    ; 
			
			\draw (593.5,879) node [anchor=north west][inner sep=0.75pt]   [align=left] [font=\Large] {$\displaystyle res \ $};

            \draw    (530,910) -- (591.5,910) ;
			\draw    (591.5,910) -- (591.5,930) ;
			\draw    (545.5,930) -- (591.5,930) ;
			\draw (593.5,913) node [anchor=north west][inner sep=0.75pt]   [align=left] [font=\Large] {Run$\displaystyle \ \mathsf{PKE.Enc}$};

            \draw [shift={(543.5,930)}, rotate = 0] [fill={rgb, 255:red, 0; green, 0; blue, 0 }  ][line width=0.08]  [draw opacity=0] (8.93,-4.29) -- (0,0) -- (8.93,4.29) -- cycle    ; 
			
            \draw   (529.5,930) -- (542.5,930) -- (542.5,963) -- (529.5,963) -- cycle ;      

            \draw  [dash pattern={on 0.84pt off 2.51pt}]  (541.5,963) -- (591.5,963) ;
			\draw  [dash pattern={on 0.84pt off 2.51pt}]  (591.5,963) -- (591.5,984) ;
			\draw  [dash pattern={on 0.84pt off 2.51pt}]  (529.5,984) -- (591.5,984) ;
   
			\draw [shift={(529.5,984)}, rotate = 0] [fill={rgb, 255:red, 0; green, 0; blue, 0 }  ][line width=0.08]  [draw opacity=0] (8.93,-4.29) -- (0,0) -- (8.93,4.29) -- cycle    ; 
			
			\draw (593.5,963) node [anchor=north west][inner sep=0.75pt]   [align=left] [font=\Large] {$\displaystyle c'_{res} \ $};

            \draw    (117,1010) -- (512.5,1010) ;
            \draw [shift={(117,1010)}, rotate = 0] [fill={rgb, 255:red, 0; green, 0; blue, 0 }  ][line width=0.08]  [draw opacity=0] (8.93,-4.29) -- (0,0) -- (8.93,4.29) -- cycle    ;
            
            \draw (196,990) node [anchor=north west][inner sep=0.75pt]  [font=\Large]  {$m_{3} =\ \langle t_{3} ,\ c'_{res},\ \sigma _{T}(\mathsf{H}( t_{3} \ || \ c'_{res}\ )) \rangle $}; 
            

            \draw    (117,1025) -- (178.5,1025) ; 
			\draw    (178.5,1025) -- (178.5,1045) ;
			\draw    (132.5,1045) -- (178.5,1045) ;
			\draw [shift={(129.5,1045)}, rotate = 0] [fill={rgb, 255:red, 0; green, 0; blue, 0 }  ][line width=0.08]  [draw opacity=0] (8.93,-4.29) -- (0,0) -- (8.93,4.29) -- cycle    ;
			\draw   (116.5,1045) -- (129.5,1045) -- (129.5,1078) -- (116.5,1078) -- cycle ;
			\draw  [dash pattern={on 0.84pt off 2.51pt}]  (128.5,1078) -- (180.5,1078) ;
			\draw  [dash pattern={on 0.84pt off 2.51pt}]  (180.5,1078) -- (180.5,1099) ;
			\draw  [dash pattern={on 0.84pt off 2.51pt}]  (120.5,1099) -- (179.5,1099) ;
			\draw [shift={(117.5,1099)}, rotate = 0] [fill={rgb, 255:red, 0; green, 0; blue, 0 }  ][line width=0.08]  [draw opacity=0] (8.93,-4.29) -- (0,0) -- (8.93,4.29) -- cycle    ;

            \draw (187,1030) node [anchor=north west][inner sep=0.75pt]   [align=left] [font=\Large] {Run$\displaystyle \ \mathsf{PKE.Dec}$};

            \draw (187,1085) node [anchor=north west][inner sep=0.75pt]   [align=left] [font=\Large] {$\displaystyle res$};
            
		\end{tikzpicture}
	}
	\caption{$\mathsf{\trustedPervPPML}$}
	\label{fig:trustedppml}
\end{figure}

\subsection{PervPPML: TEE Setting}
\label{subsec:tparty}
\noindent \textbf{\textit{High-Level Overview:}} \enskip While we deem \threePervPPMLTTT{} appropriate for a multi-client model with \textbf{A} as the model owner, it is important to acknowledge certain security vulnerabilities when considering a zero-trust environment or when contemplating the potential collusion of entities aiming to undermine our protocol. In light of these concerns, we present an enhanced protocol, \trustedPervPPMLTTT{}, a trusted variant of \threePervPPMLTTT{} specifically designed for the multi-client model that offers comprehensive trust via a \textbf{TEE}. By incorporating a TEE, we fortify the security of our system and ensure its resilience against adversarial collaboration. \trustedPervPPMLTTT{} consists of four parties: a set of users $\mathcal{U}$, \textbf{CSP}, \textbf{TEE} and an analyst \textbf{A}. The protocol's steps are similar to the steps described for \threePervPPMLTTT{} (\autoref{subsec:3party}) with a few modifications. For example, similar to the \threePervPPMLTTT{} setting, we assume that \textbf{A} is the owner of the ML model, while a user $u_i$ provides the data $x_i$. However, in \trustedPervPPMLTTT{}, the \textbf{TEE} is responsible for generating the required HHE keys ($\mathsf{pk_T, sk_T, evk_T}$). Subsequently, the \textbf{TEE} publishes $\mathsf{pk_T}$, and transfers $\mathsf{evk_T}$ to the \textbf{CSP}. The upload phase remains the same as in the \threePervPPMLTTT{} setting. For the evaluation phase, \textbf{A} requests a prediction to be performed on the stored $c'_{x_i}$ by sending a homomorphically encrypted pre-trained ML model parameters $(c_{w}, c_{b})$ to the \textbf{CSP}. On receiving this request, the \textbf{CSP} transforms $c_{x_i}$ into a homomorphic ciphertext $c'_{x_i}$ using $\mathsf{evk_T}$ and $c_{\mathsf{K_{i}}}$. The \textbf{CSP} then produces an encrypted result $c_{res}$ using $(c_{w}, c_{b})$, $c'_{x_i}$ and $\mathsf{evk_T}$ and transfers the result to the \textbf{TEE}. The \textbf{TEE} retrieves the prediction result, $res$, and then encrypts it with \textbf{A}'s public key. The resulting ciphertext is then returned to \textbf{A} for the final decryption.

\noindent\textbf{\textit{Formal Construction:}} 
\autoref{fig:trustedppml} provides a complete overview of the \trustedPervPPMLTTT{} protocol.

\noindent\framebox{$\mathsf{\mathbf{\trustedPervPPML.Setup}}$:}
Each party generates their respective signing and verification key pair for the signature scheme $\sigma$ and publishes the verification keys. 
\textbf{TEE} runs $\mathsf{HHE.KeyGen}$ to output the public, private and evaluation keys ($\mathsf{pk_{T}, sk_{T}, evk_{T}}$). On successful completion, \textbf{TEE} publishes $\mathsf{pk_{T}}$ and transfers $\mathsf{evk_{T}}$ to the \textbf{CSP}. The steps of this phase are summarized below:

\begin{myframe}{}
\begin{small}
\textbf{\underline{$\mathbf{u_i}$ computes}:}

$\mathsf{(pk_{u_{i}}, sk_{u_{i}})} \leftarrow \mathsf{PKE.Gen(1^{\lambda})}$

\medskip

\textbf{\underline{CSP computes}:}

$\mathsf{(pk_{CSP}, sk_{CSP})} \leftarrow \mathsf{PKE.Gen(1^{\lambda})}$

\medskip

\textbf{\underline{TEE computes}:}

	$\mathsf{(ver_{T}, sign_{T})} \leftarrow \mathsf{PKE.Gen(1^{\lambda})}$,\\
	$\mathsf{(pk_T, sk_T, evk_T)} \leftarrow \mathsf{HHE.KeyGen(1^{\lambda})}$ 

\medskip

\textbf{\underline{A computes}:}

	$\mathsf{(pk_{A}, sk_{A})} \leftarrow \mathsf{PKE.Gen(1^{\lambda})}$
\end{small}
\end{myframe}

\noindent\framebox{$\mathsf{\mathbf{\trustedPervPPML.Upload}}$:}
For this phase, we assume the existence of multiple users, $\mathcal{U} = \{u_1, \ldots, u_n\}$. Each $u_i$ runs $\mathsf{HHE.Enc}$ algorithm, which comprises the $\mathsf{SKE.Gen}$, $\mathsf{SKE.Enc}$, and $\mathsf{HE.Enc}$ algorithms. To proceed, $u_i$ first generates a unique symmetric key $\mathsf{K_{i}}$ using $\mathsf{SKE.Gen}$. 
Subsequently, $u_i$ encrypts their data by running $\mathsf{SKE.Enc}$ with inputs $\mathsf{K_i}$ and $x_{i}$, and outputs $c_{x_i}$. Furthermore, $u_i$ homomorphically encrypts the symmetric key $\mathsf{K_{i}}$ using $\mathsf{HE.Enc}$ with inputs $\mathsf{pk_T}$ and $\mathsf{K_i}$, and outputs $c_\mathsf{K_{i}}$. And finally, $u_i$ sends the encrypted values ($c_{x_i}$ and $c_\mathsf{K_i}$) via $m_{1}$ (\autoref{fig:trustedppml}) to the \textbf{CSP}. The steps of this phase are summarized below:

\begin{small}
\begin{gather*}
    \mathsf{K_{i}} \leftarrow \mathsf{SKE.Gen(1^{\lambda})},  \\
        \mathsf{SKE.Enc(K_i},x_{i}) \rightarrow c_{x_i},  \\
	c_\mathsf{K_{i}} \leftarrow \mathsf{HE.Enc(pk_T,\mathsf{K_i})}\\
     m_{1} = \langle t_{1}, c_{x_i}, c_{\mathsf{K_{i}}}, \sigma_{u_{i}}(\mathsf{H}(t_{1} || c_{x_i} || c_{\mathsf{K_{i}}}))\rangle
\end{gather*}    
\end{small}

\noindent\framebox{$\mathsf{\mathbf{\trustedPervPPML.Eval}}$:}
The secure evaluation phase is initiated by \textbf{A} to gain insight into the data provided by an arbitrary number of users. To begin, \textbf{A} homomorphically encrypts their ML model parameters $(w, b)$ by running  $\mathsf{HE.Enc}$ with inputs $\mathsf{pk_A}$ and $(w, b)$, and outputs encrypted model $(c_{w}, c_{b})$. \textbf{A} then sends the encrypted parameters via $m_{2}$ to the \textbf{CSP} (\autoref{fig:trustedppml}). On receiving the encrypted model, the \textbf{CSP} transforms $c_{x_i}$ into homomorphic ciphertext $c'_{x_i}$ by executing $\mathsf{HHE.Decomp}$ which takes as input $\mathsf{evk_T}, c_{x_i},$ and  $c_\mathsf{K_i}$, and outputs $c'_{x_i}$. Subsequently, \textbf{CSP} runs  $\mathsf{HHE.Eval}$ with inputs $(\mathsf{evk_{T}}, (c_{w}, c_{b}), c'_{x_i})$ to output an encrypted prediction $c_{res}$. This phase is summarized below:

\begin{myframe}{}
\begin{small}
\textbf{\underline{A computes}:}
        
	$(c_{w}, c_{b}) \leftarrow \mathsf{HE.Enc(pk_T}, (w, b))$
\medskip 
 
\textbf{\underline{CSP computes}:}

	$c'_{x_i} \leftarrow \mathsf{HHE.Decomp(evk_T}, c_{x_i}, c_\mathsf{K_i})$, \\
	$c_{res} \leftarrow \mathsf{HHE.Eval(evk_{T}}, (c_{w}, c_{b}), c'_{x_i})$
\end{small}    
\end{myframe}

\begin{equation*}
    m_{2} = \langle t_{2}, (c_{w}, c_{b}), \sigma_{A}(\mathsf{H}(t_{2} || (c_{w}, c_{b})))\rangle
\end{equation*}

\noindent\framebox{$\mathsf{\mathbf{\trustedPervPPML.Classify}}$:} 
To commence the classification phase, \textbf{CSP} transfers the encrypted prediction result, $c_{res}$, to the \textbf{TEE} for secure decryption and re-encryption. On receiving $c_{res}$, \textbf{TEE} executes $\mathsf{HHE.Dec}$ with inputs $\mathsf{sk_{T}}$ and $c_{res}$, to output the prediction results $res$. The prediction result is then encrypted with \textbf{A}'s public key using $\mathsf{PKE.Enc}$, and the resulting ciphertext $c'_{res}$ sent to \textbf{A} via $m_3$ (\autoref{fig:trustedppml}). After \textbf{A} receives the message, they will decrypt the ciphertext with their private key using $\mathsf{PKE.Dec}$ and receive the prediction result $res$. The steps of this phase are summarized below:

\begin{myframe}{}
        \begin{small}
            \textbf{\underline{TEE computes}:}

                $\mathsf{HHE.Dec(sk_T}, c_{res}) \rightarrow res$, \\
                $c'_{res} \leftarrow \mathsf{PKE.Enc(pk_A}, res)$

                \medskip
            \textbf{\underline{A computes}:}
            
                $\mathsf{PKE.Dec(sk_A}, c'_{res}) \rightarrow res$

        \end{small}        
\end{myframe}
\begin{equation*}
    m_{3} = \langle t_{3}, c'_{res}, \sigma_{T}(\mathsf{H}(t_{3} || c'_{res}))\rangle
\end{equation*}

\section{Threat Model and Security Analysis}
\label{sec:Threat&SecAnal}
In this section, we define the threat model 
used to prove the security of \protocolTTT{} by formalizing the capabilities of an adversary $\mathcal{ADV}$. The PASTA~\cite{dobraunig2021pasta} scheme we adopt as our underlying cryptographic scheme has been proven to be resilient against differential and linear statistical attacks and their variations. The cipher construction incorporates changing linear layers during encryption, which ensures defence against statistical attacks. The authors provide a rigorous proof that this layer instantiation is secure, ensuring full diffusion throughout the entire scheme, even in the best-case scenario for an attacker~\cite{dobraunig2021pasta}. PASTA is also secure against algebraic attacks, such as Linearization and Gröbner Basis Attacks. This security is achieved through proper parameter selection, including choosing prime numbers larger than 216 and selecting sufficient input words to achieve 128-bit security. These parameter choices lead to a substantial number of monomials, making solving required equations computationally complex and improbable.

Following the threat model presented in~\cite{10.1145/3605098.3635983} we assume that \trustedPervPPMLTTT{} is secure against the attacks identified in~\cite{10.1145/3605098.3635983} due to the core construction being related to \threePervPPMLTTT{}. However, we present an additional attack and prove the security of 
\threePervPPMLTTT{} and \trustedPervPPMLTTT{} against the Data Reconstruction Attack. 

\begin{myAttack}[Data Reconstruction Attack]

	Let $\mathcal{ADV}$ be a malicious adversary. $\mathcal{ADV}$ successfully launches the data reconstruction attack if she manages to compromise the \textbf{CSP} and reconstruct the users data via the inference result $res$ and the underlying ML model $f$.
\end{myAttack}

\subsection{Security Analysis}\label{subsec:SecAnal}

We now prove the security of our 
protocols in the presence of the adversary defined earlier. 

\begin{proposition}[Data Reconstruction Attack Soundness]
	Let $res$ be the inference result of a multi-layered ML model $f$, $\mathsf{HE}$ a semantically secure homomorphic encryption scheme and $\mathsf{PKE}$ an IND-CCA2 public key encryption scheme. We assume that $\mathcal{ADV}$ can corrupt the \textbf{CSP} in \texttt{PervPPML}, hence gaining access to $(c_w, c_b)$ and $c_{res}$. Then $\mathcal{ADV}$ cannot successfully launch the data reconstruction attack on \texttt{PervPPML}.

\end{proposition}
\begin{proof}
To successfully launch the Data Reconstruction attack, $\mathcal{ADV}$ has two options:
\begin{enumerate}
\item {{\color{magenta} ({\threePervPPMLTTT{}})}:} In this option of the attack, $\mathcal{ADV}$ gains access to $c_{res}$, the HE encrypted inference result, and $(c_w, c_b)$, the HE encrypted parameters of the multi-layered ML model $f$, owned by \textbf{A}. To this end, $\mathcal{ADV}$ successfully launches this attack if she successfully retrieves $f$ in the form of $(w, b)$ from $(c_w, c_b)$ and $res$ from $c_{res}$. $\mathsf{HE}$ has been proven to be semantically secure; hence, the likelihood of $\mathcal{ADV}$ successfully decrypting the encrypted inference result $c_{res}$ or the encrypted parameters $(c_w, c_b)$ is considered negligible.

\item {{\color{magenta} ({\trustedPervPPMLTTT{}})}:} In this variant of the attack, $\mathcal{ADV}$ eavesdrops on $m_{3} = \langle t_{3}, c'_{res}, \sigma_{T}(\mathsf{H}(t_{3} || c'_{res}))\rangle$ sent to \textbf{A} to obtain $c'_{res}$, the PKE encrypted inference result. As discussed previously, HE is considered semantically secure; hence, the likelihood of $\mathcal{ADV}$ successfully decrypting $c_{res}$ and $(c_w, c_b)$ is negligible. Additionally, the chosen $\mathsf{PKE}$ scheme is proven to be IND-CCA2 secure; therefore, the likelihood of $\mathcal{ADV}$ successfully decrypting $c'_{res}$ is also considered negligible.
\end{enumerate}
Due to the negligible probability of obtaining $res$ and $(w, b)$ in \texttt{PervPPML}, this attack fails.
\end{proof}


\subsection{TEE Security}
\label{subsec:teesecurity}
There are various hardware vendors that offer unique implementations of TEEs to enhance system security. These include AMD's Secure Encrypted Virtualization (SEV), Intel's Software Guard Extensions (SGX), Intel's Trusted Domain Extensions (TDX), ARM's TrustZone, and RISC-V's Keystone~\cite{tee2014} 
Cryptographic primitives play a vital role in ensuring the confidentiality and integrity of the TEE throughout its lifecycle. The attestation functionality of TEEs guarantees that they are in a known good state before loading any code. Signed binaries ensure that only approved code is loaded, while encrypted and integrity-checked data safeguard it against unauthorized access or tampering by untrusted parties~\cite{tee2015}. These measures are tied to the CPU, protecting the TEE from the main operating system, potential attackers, and even the user. Note that even the legitimate user of a system can be considered a potential adversary in some instances, such as Digital Rights Management (DRM) or specific banking operations, where subverting the system's normal operation could lead to financial gain. The separation of the TEE from the Rich Execution Environment (REE) enables more rigorous verification to be applied to the security-critical components of a system. This segregation allows for focused scrutiny and ensures that the essential parts responsible for security receive thorough evaluation and protection~\cite{tee2014}. Hence, by utilizing a \textbf{TEE} in \trustedPervPPMLTTT{}, we provide a secure and isolated region to generate HHE keys and perform HHE decryptions that are safe from even a malicious \textbf{CSP}. 

\section{Performance Analysis}
\label{sec:performanceanalysis}
In this section, we first extensively evaluate the performance of the 
\threePervPPMLTTT{} protocol, focusing on their computational and communication overhead. To provide concrete evidence of the efficiency of the protocols, we then compared the performance of it 
against a basic BFV HE scheme~\autoref{subsec:commparison}. For these evaluations, we utilized a dummy dataset where each data input was a vector of four random integers, and the weights and biases were also integer vectors of length four. In the next phase of our experiments, ~\autoref{subsec:MLApp}, we reported the process and results of building a privacy-preserving inference protocol on sensitive medical data (ECG) based on the \threePervPPMLTTT{} construction. Our primary testbed was a commercial desktop with a 12th Generation Intel i7-12700 CPU with~20 cores and~32GB of RAM running on an Ubuntu 20.04 operating system. Furthermore, in all the evaluations, we utilized the SEAL cryptographic library\footnote{\url{https://github.com/microsoft/SEAL}} for basic HE operations and PASTA library\footnote{\url{https://github.com/IAIK/hybrid-HE-framework}} to implement the secure symmetric cipher. 
The polynomial modulus was set at~16,384 for all HE operations to accommodate the complexity of the protocols operations. To ensure statistical significance, each experiment was repeated 50 times, with the average results considered to provide a comprehensive overview of the performance of each algorithm under evaluation.


\subsection{Computational Analysis}
\label{subsec:computation}
This subsection focused on the computational performance of the core algorithms executed by each entity in the 
\threePervPPMLTTT{} protocol. More precisely, we measured the time taken to execute each HHE algorithm in each protocol phase by the responsible party. For the 
{\texttt{\threePervPPML.Setup}} phase, we observed the time taken to generate a set of HHE keys $u_i$ and \textbf{A}, respectively. For implementation purposes, the $\mathsf{HHE.KeyGen}$ algorithm involved the generation of encryption parameters $\mathsf{Params}$, secret key $\mathsf{sk}$, public key $\mathsf{pk}$, relinkey $\mathsf{rk}$, and a Galois key $\mathsf{gk}$, and took~243 milliseconds to execute. Subsequently, in the 
{\texttt{\threePervPPML.Upload}} phase, we measured the time taken to homomorphically encrypt a symmetric key $\mathsf{K}$ using the $\mathsf{HE.Enc}$ algorithm and the time taken to execute the $\mathsf{SKE.Enc}$ algorithm for various data inputs ranging from~1 to~300 (each data input is an integer vector of length 4). $\mathsf{HE.Enc}$ took~7 milliseconds to run. For a single data input, $\mathsf{SKE.Enc}$ executed in~2 milliseconds and~600 milliseconds for~300 data inputs (\autoref{tab:execution}). 

When evaluating \texttt{\threePervPPML.Eval}, we first measured the performance of the $\mathsf{HE.Enc}$ algorithm at \textbf{A} and then the performance of the $\mathsf{HHE.Decomp}$ and $\mathsf{HHE.Eval}$ algorithms on various number of symmetric ciphertext inputs from~1 to~300 at \textbf{CSP}. $\mathsf{HE.Enc}$ took~16 milliseconds to execute, while the results for $\mathsf{HHE.Decomp}$ were 3595.6 seconds. 
For a single input, $\mathsf{HHE.Eval}$ took~38 milliseconds to execute and~11.6 seconds for~300 inputs.


\begin{table}[ht!]
	\centering
		\begin{tabular}{|c|c|c|c|c|}
			\hline
			\rowcolor{frenchblue}		
			\color{white}{Inputs} & \color{white}{$\mathsf{SKE.Enc}$} & \color{white}{$\mathsf{HHE.Decomp}$}  & \color{white}{$\mathsf{HHE.Eval}$} & \color{white}\textbf{$\mathsf{HHE.Dec}$} \\
			\hline 
			1 &  2 ms  &  11.9 s & 0.038 s & 3 ms\\
			\hline
			\rowcolor{Gray}
			50 &  100 ms  &  599.1 s & 1.96 s & 150 ms \\
			\hline
			100 &  200 ms  & 1197.7 s & 3.93 s & 300 ms \\
			\hline
			\rowcolor{Gray}
			150 &  300 ms  & 1794.5 s & 5.77 s & 450 ms \\
			\hline
			200 & 400 ms & 2394.2 s & 7.69 s & 600 ms\\
			\hline
			\rowcolor{Gray}
			250 & 500 ms & 2989.2 s & 9.61 s & 750 ms\\
			\hline
			300 & 600 ms & 3595.6 s & 11.61 s & 900 ms\\
			\hline
	\end{tabular}
	\caption{Computational Analysis \label{tab:execution}}
\end{table}
\noindent 
Finally, in {\texttt{\threePervPPML.Classify}} phase, we focused primarily on measuring the cost of the $\mathsf{HHE.Dec}$ algorithm for a range of homomorphic ciphertext inputs from~1 to~300. For a single input, $\mathsf{HHE.Dec}$ ran in~3 milliseconds, and for~300 inputs, it ran in~900 milliseconds (\autoref{tab:execution}). 

\subsection{Communication Cost}
\label{subsec:comm2}
To measure each protocol's communication costs, we consider the size of the data transferred between parties partaking in the protocol. Using \threePervPPMLTTT{}, evaluations of the communication costs of it were based on the following observations: On completing the {\texttt{\threePervPPML.} \texttt{Setup}} phase, \textbf{A} sends $evk_i$ to the \textbf{CSP}. Subsequently, during the {\texttt{\threePervPPML.Upload}} phase, $u_i$ sends a homomorphically encrypted symmetric key $c_\mathsf{K}$ along with symmetrically encrypted data $c = (c_1, \ldots, c_n)$. To initiate the {\texttt{\threePervPPML.Eval}} phase, \textbf{A} first sends homomorphic ciphertexts of the ML weights $(c_w, c_b)$ to the \textbf{CSP}, which is approximately~3.65 MB and is constant irrespective of the number of data inputs which is the HE weight and bias vector. Then, \textbf{CSP} sends the resulting homomorphic ciphertexts $c_{res} = (c_{res_1}, \ldots, c_{res_n})$ back to \textbf{A}. \autoref{tab:comm2} provides an overview of the communication costs incurred by \protocolTTT{} for different data inputs.

\begin{table}[ht!]
	\centering
		\begin{tabular}{|c|c|c|c|c|c}
			\hline
			\rowcolor{frenchblue}		
			\color{white}{Data Inputs} & \color{white}{$\mathsf{c = (c_1, \ldots, c_n)}$} & \color{white}{$\mathsf{c_{res} = (c_{res_1}, \ldots, c_{res_n})}$} \\
			\hline 
			1 &  56 B & 1.83 MB\\
			\hline
			\rowcolor{Gray}
			50 & 2800 B & 91.3 MB\\
			\hline
			100 &  5600 B & 182.6 MB\\
			\hline
			\rowcolor{Gray}
			150 &  8400 B & 273.9 MB\\
			\hline
			200 & 11200 B & 365.2 MB\\
			\hline
			250 & 14000 B & 456.7 MB \\
			\hline
            \rowcolor{Gray}
            300 & 16800 B & 547.9 MB \\
			\hline
	\end{tabular}
	\caption{Communication Analysis -- \protocolTTT{}  \label{tab:comm2}}
\end{table}

\subsection{Comparison with plain BFV}
\label{subsec:commparison}
To provide concrete evidence of the efficiency of both our proposed protocols, we implemented a plain BFV scheme with a similar architecture to \threePervPPMLTTT{}~and compared the results. More precisely, we measured the performance of a plain BFV scheme, where a user continuously encrypts data input homomorphically before outsourcing them to the \textbf{CSP}. The same encryption parameters were used for all implementations. For these experiments, we only focused on comparing the total computational and communication costs of running the {\texttt{\threePervPPML.Upload}} phase of our protocol, with the cost of continuously using HE encryption in the plain BFV. As with all previous experiments, we vary the number of data inputs from~1 to~300 (each data input is an integer vector of length 4).


\noindent For a single data input, \texttt{\threePervPPML.Upload} took~9 milliseconds to execute, while the plain BFV scheme took~7 milliseconds to perform a single HE encryption. It is worth noting that the plain BFV scheme is marginally faster for a single data value. However, this is due to the fact that \texttt{\threePervPPML.Upload} involves two operations (a symmetric encryption operation and an HE encryption operation), while the plain BFV scheme involves just one HE encryption operation. However, when the number of data values was increased to 300, \texttt{\threePervPPML.Upload} ran in~0.608 seconds, while the plain BFV scheme ran in~2.2 seconds. 


\noindent Subsequently, we compared the communication expenses by measuring the total size of transferable ciphertext data in bytes from a user $u_i$ to \textbf{CSP}. Overall, \texttt{\threePervPPML.Upload} sent approximately~1.8 MB of ciphertext data for both a single input and~300 inputs. This is primarily because the size of symmetric ciphertext is almost negligible as compared to that of a homomorphic ciphertext (\autoref{tab:comm2}). The plain BFV, on the other hand, sent approximately~1.82 MB of ciphertext data for a single data input and~547.8 MB for~300 data inputs.
These results show, that \protocolTTT{} reduces both the communication and computational burden of $u_i$ and transfers them to \textbf{CSP}, paving the way for potentially implementing 
HHE on 
resource-constrained devices.


\subsection{PPML Application}
\label{subsec:MLApp}
In the final phase of our evaluations, we demonstrate the real-world applicability of our \threePervPPMLTTT{} protocol by applying it to a PPML application with a sensitive heartbeat dataset to classify whether a heartbeat is subjected to heart disease or not. 
More specifically, we employed the MIT-BIH dataset~\cite{moody2001impact}, which is a dataset of human heartbeats obtained from 47 subjects from 1975 to 1979 containing 48 half-hour excerpts of two-channel ECG recordings.


\subsubsection{MIT-BIH dataset}
We used the processed ECG data from~\cite{abuadbba2020can}, which contains a train split and a test split. Both splits are comprised of 13,245 ECG examples of 128 length, with values ranging from $[0, 1]$ that belong to five classes, labeled as N, R, L, A and V. For our experiments we simplify the task to a binary classification, where N is considered a normal heartbeat and  (L, R, A, V) are called "diseased heartbeats". 
We balance the dataset to contain an equivalent amount of training samples per class and quantized the values into 4-bit integers with values ranging from $[0, 15]$ (\autoref{fig:mitbih_float_int_signals}).

To compare the accuracy of the floating-point and integer datasets we trained a simple neural network with one Fully Connected (FC) layer with a sigmoid activation function. To train a model with quantized ECG data, we employed the PocketNN framework~\cite{song2022pocketnn}. Training in floating-point arithmetic produced the best train accuracy of 89.36\% at epoch 482 and the best test accuracy of 88.93\% at epoch 500. With respect to training in integer arithmetic, we got the best train accuracy of 86.65\% at epoch 42 and the best test accuracy of 87.06\% at epoch 31. From these results, we observed that the test accuracy of the integer neural network was only 1.87\% less than the test accuracy produced by the floating-point neural network. Following these results, we aimed to test our protocol \threePervPPMLTTT{} and the impacts of HHE on the inference phase. To do so we built \ecgPPML{} -- a privacy-preserving inference protocol on the quantized integer ECG data based on the same construction.

\begin{figure}[ht!]
	\centering
	\includegraphics[width=0.45\textwidth]{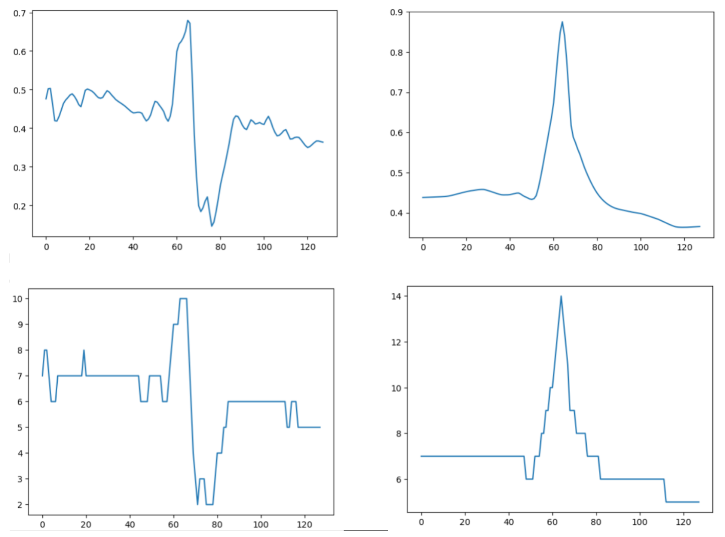}
	\caption{Quantizing ECG data into 4-bit integers. Top: Floating-point ECG data. Bottom: Corresponding quantized ECG data.}
	\label{fig:mitbih_float_int_signals}
\end{figure}

\noindent As a start, we ran the encrypted inference protocol on different numbers of examples from the test split, then compared the predictions with the ground-truth outputs to get the encrypted test accuracy. For each experiment, we also made inferences on plaintext ECG data in floating point and integer arithmetic to compare the results. 
\autoref{tab:ecgPPMLAcc} shows the encrypted and plaintext accuracies on different numbers of data input. Overall, the accuracy of plaintext inference in integer arithmetic is very similar to encrypted inference accuracies. When the number of input examples is low (1-500), integer arithmetic and encrypted accuracies are comparable or even higher than plaintext floating-point inference. When the number of data inputs increases (1000-2000 examples), plaintext integer and encrypted inference produces slightly lower accuracies (0.5-0.8\% lower). For 1000-2000 input data samples, encrypted inference had higher accuracy than plaintext integer inference but with a minimal margin (0.1-0.15\%). We reason that this is due to HE noise, which helps make a few more correct predictions.

\begin{table}[ht!]
	\centering
		\begin{tabular}{|c|c|c|c|c|c|c|}
			\hline
			\rowcolor{frenchblue}
			\color{white}{Data Inputs} & \color{white}{Plaintext (Float)} & \color{white}{Plaintext (Integer)} & \color{white}{~~Encrypted~~} \\
			\hline 
			1 &  100 \% & 100 \% & 100 \%\\
			\hline
			\rowcolor{Gray}
			10 & 90 \% & 90 \% & 90 \% \\
			\hline
			20 &  90 \% & 95 \% & 90 \%  \\
			\hline
			\rowcolor{Gray}
			50 &  88 \% & 92 \% & 90 \% \\
			\hline
			100 & 86 \% & 91 \% & 90 \% \\
			\hline
            \rowcolor{Gray}
			500 & 87 \% & 87.2 \% & 86.8 \% \\
			\hline
            1000 & 87.9 \% & 87.3 \% & 87.4 \% \\
			\hline
            \rowcolor{Gray}
            2000 & 88.2 \% & 87.4 \% & 87.55 \% \\
			\hline
	\end{tabular}
	\caption{Accuracy Analysis -- \ecgPPML{}  \label{tab:ecgPPMLAcc}}
\end{table}
\noindent Subsequently, we evaluated the computational overhead of the 
\ecgPPML{} protocol (\autoref{tab:ecgPPMLComp}). In the integer and float plaintext inference protocol, the client and analyst are not required to perform any computation as they outsource all their data and neural network model to the CSP. 
Therefore, in~\autoref{tab:ecgPPMLComp}, the client and analyst only have a single column for the encrypted inference results. Looking at these results (i.e.\ across all data input examples for the encrypted inference protocol), we observed that the CSP is responsible for the most computational overhead (99\% or even more), which increased linearly with the number of data inputs. Compared to plaintext inference, encrypted inference is more computationally expensive. However, encrypted inference for one data sample takes~12.18 seconds on a commercial desktop, which is a promising result. Furthermore, these experimental results align with our goal and expectation for the HHE protocol, as we want the client and analyst to do minimal work, and most of the computations take place in the CSP.

Finally, we analyzed the communication cost of the \ecgPPML{} protocol and reported the results in~\autoref{tab:ecgPPMLComm}. For the communication between the client and the CSP, we observed that when the number of data inputs was low (1-500), the encrypted communication cost was very high compared to the plaintext costs due to the size of the HE ciphertext of the symmetric key (1.8 Mb). However, once the number of input examples increased, the size of the symmetrically encrypted data being sent to the CSP increased linearly, similar to the plaintext size, making this difference unimportant. There is no communication between the client and the analyst in plaintext protocols, while in the encrypted inference protocol, the client only transfers the HE public key to the analyst, which has a fixed size of 2.06 Mb. Overall, the communication cost for the client was minimal and increased linearly with the number of data inputs submitted to the CSP. The communication cost between the analyst and CSP was also minimal for plaintext inference since the plaintext weights, biases, and results were small. On the other hand, the majority of the communication cost for the encrypted inference protocol was incurred between the analyst and CSP. This increased cost is caused by the HE-encrypted output of the linear layer that needs to be sent from the CSP to the analyst. Hence, if the number of data inputs increases, the communication cost between the analyst and CSP will increase linearly. We observed that the communication cost for 2000 data input examples is 5548.21 Mb, or about 5 Gb of data that needs to be transmitted, which is a reasonable result for today's internet bandwidth.

\begin{table}[ht!]
\centering

\scalebox{0.8}{
\begin{tabular}{|c|c|c|c|c|c|} 
\hline
\rowcolor{frenchblue}
\textcolor{white}{Data Inputs} & \textcolor{white}{Client}    & \textcolor{white}{Analyst}   & \multicolumn{3}{c|}{\textcolor{white}{CSP }}                                                                                                                             \\ 
\hline
\rowcolor[rgb]{0.208,0.518,0.894}                               & \textcolor{white}{Encrypted} & \textcolor{white}{Encrypted} & \multicolumn{1}{c|}{\textcolor{white}{Plaintext (float)}} & \multicolumn{1}{c|}{\textcolor{white}{Plaintext (integer)}} & \textcolor{white}{Encrypted}                 \\ 
\hline
1                                                               & 0.16                         & 0.52                          & 0                                                           & 0.23                                                        & 12.18                                        \\ 
\hline
\rowcolor{Gray} 10                            & 0.18                         & 0.58                         & 0                                                           & 0.24                                                        & 120.44   \\ 
\hline
20                                                              & 0.2                         & 0.63                         & 0                                                           & 0.23                                                        & 241.21                                       \\ 
\hline
\rowcolor{Gray} 50                            & 0.27                         & 0.803                         & 0                                                           & 0.25                                                        & 601.95                                       \\ 
\hline
100                                                             & 0.4                         & 1.091                         & 0                                                           & 0.24                                                        & 1212.38                                      \\ 
\hline
\rowcolor{Gray} 500                           & 1.31                         & 3.38                         & 0.05                                                        & 0.24                                                        & \textcolor[rgb]{0.051,0.051,0.051}{6021.98}  \\ 
\hline
1000                                                            & 2.46                         & 6.2                         & 0.1                                                         & 0.25                                                        & \textcolor[rgb]{0.051,0.051,0.051}{12058.2}  \\ 
\hline
\rowcolor{Gray} 2000                          & 4.8                        & 11.92                        & 0.2                                                         & 0.27                                                        & \textcolor[rgb]{0.051,0.051,0.051}{24153.5}  \\ 
\hline
\multicolumn{1}{l}{}                                            & \multicolumn{1}{l}{}         & \multicolumn{1}{l}{}         & \multicolumn{1}{l}{}                                        & \multicolumn{1}{l}{}                                        & \multicolumn{1}{l}{}                        
\end{tabular}}
\caption{Computation Analysis -- \ecgPPML{}. All values in seconds.
\label{tab:ecgPPMLComp}}
\end{table}

\noindent The above experimental results show that our \ecgPPML{} protocol produces comparable results in accuracy compared to inference on plaintext data. Furthermore, the CSP is responsible for most of the computation costs, and the majority of communication cost also occurs between the CSP and the analyst. Meanwhile, the client is responsible for a very small amount of computation as well as communication. These results align with our vision of using HHE for PPML applications and show immense potential for HHE when applied in real-world PPML applications. 

\begin{table}[!ht]
\centering
\scalebox{0.7}{
\begin{tabular}{|c|c|c|c|c|c|c|c|}
\hline
\rowcolor{frenchblue} \textcolor{white}{Data Inputs} & \multicolumn{3}{c|}{\textcolor{white}{Client - CSP}}                                                                                                                                                                                                                                                                       & \textcolor{white}{Client - Analyst} & \multicolumn{3}{c|}{\textcolor{white}{Analyst - CSP }}                                                                                 \\ 
\hline
\rowcolor[rgb]{0.208,0.518,0.894}                               & \begin{tabular}[c]{@{}>{\cellcolor[rgb]{0.208,0.518,0.894}}c@{}}\textcolor{white}{Plaintext}\\\textcolor{white}{(float)}\end{tabular} & \begin{tabular}[c]{@{}>{\cellcolor[rgb]{0.208,0.518,0.894}}c@{}}\textcolor{white}{Plaintext}\\\textcolor{white}{(integer)}\end{tabular} & \textcolor{white}{Encrypted}             & \textcolor{white}{Encrypted}        & \begin{tabular}[c]{@{}>{\cellcolor[rgb]{0.208,0.518,0.894}}c@{}}\textcolor{white}{Plaintext}\\\textcolor{white}{(float)}\end{tabular} & \begin{tabular}[c]{@{}>{\cellcolor[rgb]{0.208,0.518,0.894}}c@{}}\textcolor{white}{Plaintext}\\\textcolor{white}{(integer)}\end{tabular} & \textcolor{white}{Encrypted}                 \\ 
\hline
1                                                               & \textcolor[rgb]{0.024,0.024,0.024}{0.0002}                                                                                            & \textcolor[rgb]{0.024,0.024,0.024}{0.0002}                                                                                              & \textcolor[rgb]{0.051,0.051,0.051}{1.8} & 2.06                                & \textcolor[rgb]{0.024,0.024,0.024}{0.0017}                                                                                            & \textcolor[rgb]{0.024,0.024,0.024}{0.000734}                                                                                            & 72.46                                        \\ 
\hline
\rowcolor{Gray} 10                            & \textcolor[rgb]{0.024,0.024,0.024}{0.002}                                                                                             & \textcolor[rgb]{0.024,0.024,0.024}{0.002}                                                                                               & 1.8                                     & 2.06                                & \textcolor[rgb]{0.024,0.024,0.024}{0.0017}                                                                                            & \textcolor[rgb]{0.024,0.024,0.024}{0.000734}                                                                                            & 97.11                                        \\ 
\hline
20                                                              & \textcolor[rgb]{0.024,0.024,0.024}{0.005}                                                                                             & \textcolor[rgb]{0.024,0.024,0.024}{0.005}                                                                                               & 1.81                                     & 2.06                                & \textcolor[rgb]{0.024,0.024,0.024}{0.0017}                                                                                            & \textcolor[rgb]{0.024,0.024,0.024}{0.000734}                                                                                            & 124.51                                       \\ 
\hline
\rowcolor{Gray} 50                            & \textcolor[rgb]{0.024,0.024,0.024}{0.012}                                                                                             & \textcolor[rgb]{0.024,0.024,0.024}{0.012}                                                                                               & 1.81                                     & 2.06                                & \textcolor[rgb]{0.024,0.024,0.024}{0.0017}                                                                                            & \textcolor[rgb]{0.024,0.024,0.024}{0.000734}                                                                                            & 206.692                                      \\ 
\hline
100                                                             & \textcolor[rgb]{0.024,0.024,0.024}{0.029}                                                                                             & \textcolor[rgb]{0.024,0.024,0.024}{0.029}                                                                                               & 1.83                                     & 2.06                                & \textcolor[rgb]{0.024,0.024,0.024}{0.0017}                                                                                            & \textcolor[rgb]{0.024,0.024,0.024}{0.000734}                                                                                            & 343.643                                      \\ 
\hline
\rowcolor{Gray} 500                           & 1.1                                                                                                                                   & \textcolor[rgb]{0.024,0.024,0.024}{1.1}                                                                                                 & \textcolor[rgb]{0.051,0.051,0.051}{2.9} & 2.06                                & \textcolor[rgb]{0.024,0.024,0.024}{0.0017}                                                                                            & \textcolor[rgb]{0.024,0.024,0.024}{0.000734}                                                                                            & \textcolor[rgb]{0.051,0.051,0.051}{1439.27}  \\ 
\hline
1000                                                            & 2.3                                                                                                                                   & \textcolor[rgb]{0.024,0.024,0.024}{2.3}                                                                                                 & \textcolor[rgb]{0.051,0.051,0.051}{4.1} & 2.06                                & \textcolor[rgb]{0.024,0.024,0.024}{0.0017}                                                                                            & \textcolor[rgb]{0.024,0.024,0.024}{0.000734}                                                                                            & \textcolor[rgb]{0.051,0.051,0.051}{2809.02}  \\ 
\hline
\rowcolor{Gray} 2000                          & 4.6                                                                                                                                   & \textcolor[rgb]{0.024,0.024,0.024}{4.6}                                                                                                 & \textcolor[rgb]{0.051,0.051,0.051}{6.4} & 2.06                                & \textcolor[rgb]{0.024,0.024,0.024}{0.0017}                                                                                            & \textcolor[rgb]{0.024,0.024,0.024}{0.000734}                                                                                            & \textcolor[rgb]{0.051,0.051,0.051}{5548.21}  \\ 
\hline
\multicolumn{1}{l}{}                                            & \multicolumn{1}{l}{}                                                                                                                  & \multicolumn{1}{l}{}                                                                                                                    & \multicolumn{1}{l}{}                     & \multicolumn{1}{l}{}                & \multicolumn{1}{l}{}                                                                                                                  & \multicolumn{1}{l}{}                                                                                                                    & \multicolumn{1}{l}{}                        
\end{tabular}
}
\caption{Communication Analysis -- \ecgPPML{}. All values in Megabytes. \label{tab:ecgPPMLComm}}
\end{table}
\noindent \textbf{\textit{Open Science \& Reproducible Research}}
To support open science and reproducible research, and provide other researchers with the opportunity to use, test, and hopefully extend our work, the source codes used for the evaluations have been anonymized and made available online\footnote{\url{https://github.com/iammrgenie/hhe_ppml}}\textsuperscript{,}\footnote{\url{https://github.com/khoaguin/PocketHHE}}.

\section{Conclusion}
\label{sec:conclusion}
This paper is one of the first attempts at using the novel concept of HHE 
effectively to address the problem of privacy-preserving machine learning. 
We have provided a realistic solution that carefully considers the vagaries of PPML. The designed approach is able to carefully balance between ML functionality and privacy so as to allow the use of PPML techniques in a wide range of areas such as pervasive computing, where, in many cases, the underlying infrastructure presents certain inbuilt limitations. By using HHE, we managed to overcome the main difficulties of PPML application in real-life scenarios, where the majority of data is collected and processed by constraint devices. Certain that the future of cryptography goes hand in hand with ML, we believe we have made the first step towards implementing PPML services with strong security guarantees, which operate efficiently in a wide range of architectures. 

\section*{Acknowledgment}
This work was funded by the HARPOCRATES EU research project (No. 101069535).

\bibliographystyle{ieeetr}
\balance
\bibliography{pervPPML}

\begin{thebibliography}{10}

\bibitem{rivest1978data}
R.~L. Rivest, L.~Adleman, M.~L. Dertouzos, {\em et~al.}, ``On data banks and privacy homomorphisms,'' {\em Foundations of secure computation}, vol.~4, no.~11, pp.~169--180, 1978.

\bibitem{gentry2009fully}
C.~Gentry, {\em A fully homomorphic encryption scheme}.
\newblock Stanford university, 2009.

\bibitem{dobraunig2021pasta}
C.~Dobraunig, L.~Grassi, L.~Helminger, C.~Rechberger, M.~Schofnegger, and R.~Walch, ``Pasta: a case for hybrid homomorphic encryption,'' {\em Transaction on Cryptographic Hardware and Embedded Systems 2023 Issue 3}, 2023.

\bibitem{bakas2019modern}
A.~Bakas and A.~Michalas, ``Modern family: A revocable hybrid encryption scheme based on attribute-based encryption, symmetric searchable encryption and sgx,'' in {\em International conference on security and privacy in communication systems}, pp.~472--486, Springer, 2019.

\bibitem{cho2021transciphering}
J.~Cho, J.~Ha, S.~Kim, B.~Lee, J.~Lee, J.~Lee, D.~Moon, and H.~Yoon, ``Transciphering framework for approximate homomorphic encryption,'' in {\em International Conference on the Theory and Application of Cryptology and Information Security}, pp.~640--669, Springer, 2021.

\bibitem{ha2022rubato}
J.~Ha, S.~Kim, B.~Lee, J.~Lee, and M.~Son, ``Rubato: Noisy ciphers for approximate homomorphic encryption,'' in {\em Advances in Cryptology--EUROCRYPT 2022}, pp.~581--610, Springer, 2022.

\bibitem{cosseron2022towards}
O.~Cosseron, C.~Hoffmann, P.~M\'{e}aux, and F.-X. Standaert, ``Towards case-optimized hybrid homomorphic encryption: Featuring the elisabeth stream cipher,'' in {\em Advances in Cryptology – ASIACRYPT 2022}, p.~32–67, Springer-Verlag, 2023.

\bibitem{10.1145/3605098.3635983}
E.~Frimpong, K.~Nguyen, M.~Budzys, T.~Khan, and A.~Michalas, ``Guardml: Efficient privacy-preserving machine learning services through hybrid homomorphic encryption,'' in {\em Proceedings of the 39th ACM/SIGAPP Symposium on Applied Computing}, SAC '24, p.~953–962, 2024.

\bibitem{khan2024learning}
T.~Khan and A.~Michalas, ``Learning in the dark: Privacy-preserving machine learning using function approximation,'' IEEE International Conference on Trust, Security and Privacy in Computing and Communications, pp.~62--71, IEEE, 2024.

\bibitem{khan2021blind}
T.~Khan, A.~Bakas, and A.~Michalas, ``Blind faith: Privacy-preserving machine learning using function approximation,'' in {\em 2021 IEEE Symposium on Computers and Communications (ISCC)}, pp.~1--7, IEEE, 2021.

\bibitem{khan2023love}
T.~Khan, K.~Nguyen, A.~Michalas, and A.~Bakas, ``Love or hate? share or split? privacy-preserving training using split learning and homomorphic encryption,'' in {\em 2023 20th Annual International Conference on Privacy, Security and Trust (PST)}, pp.~1--7, IEEE, 2023.

\bibitem{khan2023more}
T.~Khan, K.~Nguyen, and A.~Michalas, ``A more secure split: Enhancing the security of privacy-preserving split learning,'' in {\em Nordic Conference on Secure IT Systems}, pp.~307--329, Springer, 2023.

\bibitem{khan2023split}
T.~Khan, K.~Nguyen, and A.~Michalas, ``Split ways: Privacy-preserving training of encrypted data using split learning,'' in {\em 2023 Workshops of the EDBT/ICDT Joint Conference, EDBT/ICDT-WS 2023 Ioannina 28 March 2023}, vol.~3379, CEUR-WS, 2023.

\bibitem{chillotti2020tfhe}
I.~Chillotti, N.~Gama, M.~Georgieva, and M.~Izabach{\`e}ne, ``Tfhe: fast fully homomorphic encryption over the torus,'' {\em Journal of Cryptology}, vol.~33, no.~1, pp.~34--91, 2020.

\bibitem{fan2012somewhat}
J.~Fan and F.~Vercauteren, ``Somewhat practical fully homomorphic encryption,'' {\em Cryptology ePrint Archive}, 2012.

\bibitem{cheon2017homomorphic}
J.~H. Cheon, A.~Kim, M.~Kim, and Y.~Song, ``Homomorphic encryption for arithmetic of approximate numbers,'' in {\em International Conference on the Theory and Application of Cryptology and Information Security}, pp.~409--437, Springer, 2017.

\bibitem{sanyal2018tapas}
A.~Sanyal, M.~Kusner, A.~Gascon, and V.~Kanade, ``Tapas: Tricks to accelerate (encrypted) prediction as a service,'' in {\em International Conference on Machine Learning}, pp.~4490--4499, PMLR, 2018.

\bibitem{bourse2018fast}
F.~Bourse, M.~Minelli, M.~Minihold, and P.~Paillier, ``Fast homomorphic evaluation of deep discretized neural networks,'' in {\em Annual International Cryptology Conference}, Springer, 2018.

\bibitem{lou2020glyph}
Q.~Lou, B.~Feng, G.~Charles~Fox, and L.~Jiang, ``Glyph: Fast and accurately training deep neural networks on encrypted data,'' {\em Advances in Neural Information Processing Systems}, vol.~33, pp.~9193--9202, 2020.

\bibitem{sav2020poseidon}
S.~Sav, A.~Pyrgelis, J.~R. Troncoso-Pastoriza, D.~Froelicher, J.-P. Bossuat, J.~S. Sousa, and J.-P. Hubaux, ``Poseidon: privacy-preserving federated neural network learning,'' in {\em Network and Distributed System Security Symposium, {NDSS} 2021}, 2021.

\bibitem{al2020towards}
A.~Al~Badawi, C.~Jin, J.~Lin, C.~F. Mun, S.~J. Jie, B.~H.~M. Tan, X.~Nan, K.~M.~M. Aung, and V.~R. Chandrasekhar, ``Towards the alexnet moment for homomorphic encryption: Hcnn, the first homomorphic cnn on encrypted data with gpus,'' {\em IEEE Transactions on Emerging Topics in Computing}, vol.~9, no.~3, pp.~1330--1343, 2020.

\bibitem{khan2024wildest}
T.~Khan, M.~Budzys, K.~Nguyen, and A.~Michalas, ``Sok: Wildest dreams: Reproducible research in privacy-preserving neural network training,'' {\em Proceedings on Privacy Enhancing Technologies}, vol.~2024, no.~3, 2024.

\bibitem{chillotti2016faster}
I.~Chillotti, N.~Gama, M.~Georgieva, and M.~Izabachene, ``Faster fully homomorphic encryption: Bootstrapping in less than 0.1 seconds,'' in {\em Advances in Cryptology--ASIACRYPT 2016}, pp.~3--33, Springer, 2016.

\bibitem{brakerski2012fully}
Z.~Brakerski, ``Fully homomorphic encryption without modulus switching from classical gapsvp,'' in {\em Annual Cryptology Conference}, pp.~868--886, Springer, 2012.

\bibitem{gentry2012homomorphic}
C.~Gentry, S.~Halevi, and N.~P. Smart, ``Homomorphic evaluation of the aes circuit,'' in {\em Annual Cryptology Conference}, Springer, 2012.

\bibitem{bakas2022symmetrical}
A.~Bakas, E.~Frimpong, and A.~Michalas, ``Symmetrical disguise: Realizing homomorphic encryption services from symmetric primitives,'' in {\em International Conference on Security and Privacy in Communication Systems}, pp.~353--370, Springer, 2022.

\bibitem{canteaut2018stream}
A.~Canteaut, S.~Carpov, C.~Fontaine, T.~Lepoint, M.~Naya-Plasencia, P.~Paillier, and R.~Sirdey, ``Stream ciphers: A practical solution for efficient homomorphic-ciphertext compression,'' {\em Journal of Cryptology}, vol.~31, no.~3, pp.~885--916, 2018.

\bibitem{meaux2019improved}
P.~M{\'e}aux, C.~Carlet, A.~Journault, and F.-X. Standaert, ``Improved filter permutators for efficient fhe: Better instances and implementations,'' in {\em International Conference on Cryptology in India}, Springer, 2019.

\bibitem{tee2015}
M.~Sabt, M.~Achemlal, and A.~Bouabdallah, ``Trusted execution environment: What it is, and what it is not,'' {\em 2015 IEEE Trustcom/BigDataSE/ISPA}, 2015.

\bibitem{paju2023sok}
A.~Paju, M.~O. Javed, J.~Nurmi, J.~Savim{\"a}ki, B.~McGillion, and B.~B. Brumley, ``Sok: A systematic review of tee usage for developing trusted applications,'' in {\em Proceedings of the 18th International Conference on Availability, Reliability and Security}, pp.~1--15, 2023.

\bibitem{tee2014}
G.~Arfaoui, S.~Gharout, and J.~Traore, ``Trusted execution environments: A look under the hood,'' {\em 2014 2nd IEEE International Conference on Mobile Cloud Computing, Services, and Engineering}, 2014.

\bibitem{moody2001impact}
G.~B. Moody and R.~G. Mark, ``The impact of the mit-bih arrhythmia database,'' {\em IEEE engineering in medicine and biology magazine}, vol.~20, no.~3, pp.~45--50, 2001.

\bibitem{abuadbba2020can}
S.~Abuadbba, K.~Kim, M.~Kim, C.~Thapa, S.~A. Camtepe, Y.~Gao, H.~Kim, and S.~Nepal, ``Can we use split learning on 1d cnn models for privacy preserving training?,'' in {\em Proceedings of the 15th ACM Asia Conference on Computer and Communications Security}, 2020.

\bibitem{song2022pocketnn}
J.~Song and F.~Lin, ``Pocketnn: Integer-only training and inference of neural networks via direct feedback alignment and pocket activations in pure c++,'' {\em Proceedings of tinyML Research Symposium}, 2022.

\end{thebibliography}

\end{document}